\documentclass[11pt]{article}
\usepackage{latexsym}
\usepackage{amsfonts}
\usepackage{amsmath}

\usepackage{amsmath}
\usepackage{amssymb}
\usepackage{amsthm}
\usepackage{amsfonts}

\numberwithin{equation}{section}
\theoremstyle{plain}

\newtheorem{definition}{Definition}[section]
\newtheorem{theorem}{Theorem}[section]
\newtheorem{lemma}{Lemma}[section]
\newtheorem{corollary}{Corollary}[section]
\newtheorem{proposition}{Proposition}[section]
\newtheorem{remark}{Remark}

\numberwithin{lemma}{section}
\numberwithin{proposition}{section}
\numberwithin{corollary}{section}

\def \R{I\!\!R}
\def \E{I\!\!E}

\def \N{I\!\!N}

\def\INTER{\mathop{\rm {\cap}}\limits}

\begin{document}

\title{Risk Measuring under Model Uncertainty} 
\author{Jocelyne BION-NADAL
\footnote{tel: 33 1 69 33 46 25,  fax: 33 1 69 33 30 11, e-mail : jocelyne.bion-nadal@cmap.polytechnique.fr}\\ \small \it  UMR 7641  CNRS CMAP
Ecole 
 Polytechnique, 91128 Palaiseau Cedex, France\\  and\\   Magali KERVAREC\footnote{ e-mail :  mkervare@univ-evry.fr}\\
\small \it 
Laboratoire analyse et probabilit\'es, Universit\'e d' Evry, Bd F. Mitterrand, 91000 Evry France}
\date{} 

 \maketitle



\begin{abstract}
The framework of this paper is that of risk measuring under uncertainty, which is when no reference probability measure is given. To every  regular  convex  risk measure  on ${\cal C}_b(\Omega)$, we associate  a unique  equivalence class of probability measures on Borel sets, characterizing the riskless non positive elements of ${\cal C}_b(\Omega)$. We  prove that the convex risk measure  has a dual representation with a countable set  of probability measures  absolutely continuous with respect to a certain  probability measure in this  class. 
 To get these results we  study  the topological properties of the dual of the Banach space $L^1(c)$ associated to a capacity $c$.  \\ 
As application we obtain that every $G$-expectation $\E$ 
has a representation with a countable set of probability measures  absolutely continuous with respect to a  probability measure  $P$ such that $P(|f|)=0$ iff $\E(|f|)=0$.
We also apply our results to the case of 
uncertain volatility. 
\end{abstract}
{\it \small  classification:} {\small 46A20, 91B30, 46E05}


 \section{Introduction}
The purpose of this paper is to introduce a very general framework enabling the study of risk measures and dynamic risk measures in a context of model uncertainty, which is when no reference probability measure is given.\\
In order to quantify the risk in finance, Artzner et al  \cite{Artzner} have introduced  the notion of coherent (i.e. sublinear) risk measure in the context of finite probability spaces. This notion has  been extended to general probability spaces \cite{D0} and then to the convex case (\cite{FS04} and \cite{F}). The notion of 
conditional risk measure has been considered in  \cite{DSc} and \cite{BN01}.  Dynamic risk measures have then been studied in many papers, among them \cite{D}, \cite{CDK}  \cite{KS}  \cite{BN02} \cite{JBN} \cite{RS}. For the particular case of dynamic risk measures on a Brownian filtration one can cite \cite{Peng} \cite{BEK}, \cite{DPR}.
 Notice that in  all these papers on dynamic risk measures, a  reference 
probability space is fixed. This framework is  rich enough to study models with stochastic volatility or models with jumps, but not to deal with model uncertainty.\\

What means uncertainty? Usually in mathematical finance, in order to compute 
 the risk or the price associated to financial assets, one assumes that a 
reference family of liquid 
assets is given, and that 
the dynamics of these reference assets is known.  However in a context of model uncertainty the dynamics of the liquid reference assets is only assumed to belong to a certain class of models. A simple example is given, within the Brownian framework, by a class  of models with uncertain volatility. 
That is, one considers a family of possible models of the form  $dX^{\sigma}_t=b_t X^{\sigma}_t dt+\sigma_t X^{\sigma}_t dW_t$ where  $\sigma_t$ is allowed to vary inside an interval $[\underline \sigma,\overline \sigma]$. When $\sigma$ describes the set of predictable processes varying inside this interval, the 
 laws of the processes $X^t_{\sigma}$  are not  all absolutely continuous with respect to some probability measure.  Avellaneda et al \cite{AP},  Denis and Martini \cite{DM} and Denis et al \cite{DHP} have  considered the problem of pricing  for this family of models. 
Only few papers study convex risk measures in a context of uncertainty.  F\"ollmer and Schied \cite{FS04} have studied static risk measures defined on the vector space of all bounded measurable maps. This has been extended by Bion-Nadal to the conditional case in \cite{BN01}. Kervarec \cite{K} has studied static risk measures when model uncertainty is specified by a non dominated weakly compact set of probability measures.\\
In this paper, motivated by the general  context of model uncertainty, we study  regular convex risk measures defined on ${\cal C}_b(\Omega)$, the set of continuous bounded functions on a  Polish space $\Omega$. Regularity is here equivalent to continuity  with respect to a certain   capacity $c$. Considering the completion ${\cal L}^1(c)$ of
${\cal C}_b(\Omega)$
 with respect to the capacity $c$, this means that we study convex risk measures on the Banach space $L^1(c)$. 
Our main result is that for every regular convex risk measure on ${\cal C}_b(\Omega)$,   there is a unique  equivalence class of probability measures characterizing the riskless non positive elements of ${\cal C}_b(\Omega)$, and that
the convex risk measure  has a dual representation with a countable set  of probability measures all absolutely continuous with respect to a certain  probability measure belonging to this  equivalence class. The tools of the proof are the capacities, topological properties of the dual of the Banach space $L^1(c)$ associated to a capacity $c$, and convex duality for locally convex spaces.\\

The paper is organized as follows.
First, Section \ref{sec1}, we study the topological properties of the dual of $L^1(c)$. We prove that the non negative part of the dual ball of $L^1(c)$ is metric compact for the weak* topology 
$\sigma( L^1(c)^*, L^1(c))$. \\
 Section \ref{L1} deals with convex risk measures on $L^1(c)$. We prove that they satisfy the following representation formula:
\begin{eqnarray}\label{rep}
\rho\left( X\right)= \sup_{Q\in\mathcal{P}'}\left(E_Q\left[ -X\right]-\alpha\left( Q\right) \right)
\end{eqnarray}
where $\mathcal{P}'$ is a set of probability measures belonging to the dual of $ L^1(c)$. 
There are two important results in this Section. The first one is the characterization of   convex risk measures on $L^1(c)$ admitting a representation of the form (\ref{rep}) having a  compact set ${\cal P}'$ of probability measures (for the weak* topology $\sigma(L^1(c)^*, L^1(c))$). In this case, the supremum in (\ref{rep}) is a maximum. Moreover, making use of the topological results of Section \ref{sec1}, we prove that every convex risk measure on $L^1(c)$ has a dual representation  of the form (\ref{rep}) with a countable set of probability measures.\\
In section \ref{sec2} we assume that the capacity is defined on ${\cal C}_b(\Omega)$ by  $c_{p,{\cal P}}(f)= \sup_{P \in {\cal P}} E_P(|f|^p)^{\frac{1}{p}}$ for some  weakly relatively compact set ${\cal P}$ of probability measures. We prove that the capacity  $c_{p,{\cal P}}$ is equal to the capacity $c_{p,{\cal Q}}$   defined  using a certain  countable subset ${\cal Q}$ of ${\cal P}$. We introduce a new equivalence relation on the set of non negative measures belonging to the dual of $L^1(c_{p,{\cal P}})$. When ${\cal P}$ is a singleton, it  coincides with the usual equivalence relation on non negative measures. The main result of Section \ref{sec2} is the  existence of an equivalence  class of probability measures  characterizing the null elements of $L^1(c_{p,{\cal P}})_+$, that is $P$ belongs to this equivalence class if and only if  for all $f$ in $L^1(c_{p,{\cal P}})$, $(E_P(|f|)=0)  \Longleftrightarrow ( c_{p,{\cal P}}(|f|)=0)$.\\
Section \ref{secreg} deals with uniformly regular convex risk measures on ${\cal C}_b(\Omega)$. We prove that every such   risk measure on ${\cal C}_b(\Omega)$ extends into a convex risk measure on $L^1(c)$ for a certain  capacity $c$ associated to a weakly compact set ${\cal P}$ of probability measures: $c(f)=\sup_{P \in {\cal P}} E_P(f)$. Therefore we can make use of the results obtained in Sections \ref{sec2} and \ref{L1} in order to get the main result of the paper  in Theorem \ref{thmerepcont}: to every uniformly regular  convex  risk measure $\rho$ on ${\cal C}_b(\Omega)$, one can associate  a unique  equivalence class of probability measures defined on the Borel sets, called $c_{\rho}$-class, characterizing the non positive elements of ${\cal C}_b(\Omega)$ with risk $0$. The convex risk measure  has then a dual representation with a countable set  of probability measures all absolutely continuous with respect to a certain  probability measure belonging to this  $c_{\rho}$-class. \\
Section \ref{Exemples} deals with two examples. The first one  is $G$-expectations introduced by Peng \cite{P1}. The capacity associated to a $G$-expectation $\E$ is $c(f)=\E(|f|)$. As application of our results we obtain that there is a unique equivalence  class of probability  measures characterizing the non negative elements $f$ of ${\cal C}_b(\Omega)$ such that $\E(f)=0$. The  G-expectation $\E$ has then a representation in terms of a countable  set  of probability measures all absolutely continuous with respect to a certain  probability measure belonging to this class,
\begin{equation}
\E(X)=\sup_{n \in \N}E_{Q_n}(X)
\label{eqP}
\end{equation}
The second example, for which all our results apply,  is the case where model uncertainty is characterized by a relatively weakly  compact set of probability measures ${\cal P}$.

\section{Topological properties of the dual space of $L^1(c)$}
\label{sec1}
\subsection{The ordered space  $L^1(c)$}
\label{sec1.1}
Let $\Omega$  be a metrizable and separable space. One classical example, furthermore a Polish space,  is $\Omega={\cal C}_0([0,\infty[, \R^d)$ endowed with the topology of uniform convergence on compact subspaces. 
${\cal B}(\Omega)$ denotes the Borel $\sigma$-algebra on $\Omega$.
Denote ${\cal M}(\Omega)$ the set of all bounded signed measures on $(\Omega,{\cal B}(\Omega))$, and
 ${\cal M}_+(\Omega)$ the subset of  non-negative finite measures.\\
In the following  ${\cal L}$ denotes a linear vector subspace of ${\cal C}_b(\Omega)$ containing the constants, generating the topology of $\Omega$ and which is a vector lattice.                                                                                                 
Recall the following definition of a capacity.
\begin{definition} a capacity on  ${\cal L}$ 
is a semi norm $c$ defined on  ${\cal L}$ satisfying the following properties:
\begin{enumerate}
\item monotonicity: $\forall, f,g \in {\cal L}$ such that $|f| \leq |g|$, \;\; $c(f) \leq c(g)$
\item regularity along sequences:
for every sequence $f_n \in {\cal L}$ decreasing to $0$, $\inf c(f_n)=0$  
\end{enumerate}
\label{defcap11}
\end{definition}

The semi-norm $c$ is extended as in \cite{FP} Section 2 to all  real functions defined on $\Omega$:
\begin{equation}
\forall f  \;l.s.c.\;f \geq 0, \;\; c(f)=\sup\{c(\phi)|0 \leq \phi \leq f, \; \phi \in {\cal L}\} 
\label{eqsci}
\end{equation}
\begin{equation}
\forall g,\; \; c(g)=\inf\{c(f)|\; f \geq |g|, \;  f\;l.s.c.\}
\label{eqsci2}
\end{equation}
where l.s.c. means lower semi-continuous.
 ${\cal L}^1(c)$ denotes the closure of ${\cal L}$ in the set $\{g |\;c(g)<\infty\}$. From Proposition 10 of \cite{FP},   ${\cal L}^1(c)$ contains ${\cal C}_b(\Omega)$. Let $L^1(c)$ be the quotient of ${\cal L}^1(c)$ by the $c$ null elements. It is a Banach space. The following result shows that $c(1_A)$ can be expressed as the limit of a monotone sequence $c(f_n)$ for continuous functions $f_n$ with limit $1_A$, as soon as  $A$ is either an open subset or a closed subset of $\Omega$.
\begin{proposition}
\begin{itemize}
\item Let $V$ be an open subset of $\Omega$. There is an increasing sequence of non negative  continuous functions $h_n$ on $\Omega$ such that $1_V=\lim_{n \rightarrow \infty} h_n$ and $c(1_V)=\lim_{n \rightarrow \infty} c(h_n)$.
\item Let $F$ be a closed subset of $\Omega$. There is a decreasing sequence of continuous functions $g_n \leq 1$ on $\Omega$ such that $1_F=\lim_{n \rightarrow \infty} g_n$ and $c(1_F)=\lim_{n \rightarrow \infty} c(g_n)$.
\end{itemize}
\label{propopenclosed}
\end{proposition}
\begin{proof}
\begin{itemize}
\item $1_V$ is a non negative bounded l.s.c. function. Thus it  is the limit of an increasing   sequence of non negative continuous functions $f_n$. On the other hand from definition of $c(1_V)$ (equation (\ref{eqsci})), there is a sequence of continuous functions $g_n \leq 1_V$ such that $c(1_V)= \lim c(g_n)$. Let $h_1=g_1$ and for every $n$, $h_{n+1}=\sup(h_n, f_n,g_n)$. $h_n$ is an increasing sequence of continuous functions  with limit $1_V$ and such that $c(1_V)=\lim c(h_n)$. 
\item Let $F$ be a closed subset of $\Omega$. 
By definition of the capacity, $c(1_F)= \inf_{\{\psi\; l.s.c., 1_F \leq \psi\}}c(\psi)$. The infimum of two l.s.c. functions is also l.s.c. , thus there is  a decreasing sequence $\psi_n$ greater or equal to   $1_F$  such that $c(1_F)= \lim c(\psi_n)$.
Thus there is a strictly increasing sequence  $k(n)$  such that for all $n$,  $c(\psi_{k(n)}) \leq c(1_F)+ \frac{1}{n^2}$. Let $\epsilon_n>0$ such that $(\frac{1}{1-\epsilon_n})(c(1_F)+ \frac{1}{n^2}) \leq c(1_F)+ \frac{1}{n}$.
Let $V_n=\{x |\psi_{k(n)}(x)> 1- \epsilon_n\} \INTER \{ x\in\Omega ;  dist (x,F)<\frac{1}{n}\}$.  As $\psi_{k(n)}$ is
 l.s.c., $V_n$ is an open set, furthermore  $F= \cap_{n \in \N^*} V_n$. For every $n$, there is  a continuous function $f_n$ such that $F \prec f_n \prec V_n$.  One can thus construct a decreasing sequence of continuous functions $g_n$ such that $1_F \leq g_n \leq  1_{V_n}$.  Thus  the sequence $g_n$ is decreasing to $1_F$.
 As $c(1_{V_n}) \leq \frac{1}{1-\epsilon_n}  c(\psi_{k(n)}) \leq c(1_F)+ \frac{1}{n}$, it follows that 
  $c(1_F) \leq c(g_n)\leq c(1_F)+\frac{1}{n}$.
 \end{itemize}
\end{proof}

 Further definitions and results on capacities are recalled in the Appendix (Section \ref{secap}). We refer also to  \cite{FP}. 
 \subsection*{Partial order on $L^1(c)$}
\begin{definition} Let $X \in L^1(c)$. We say that $X \geq 0$ if there is a sequence $(f_n)_{n \in \N}$, $f_n \in {\cal L}, \;f_n\geq 0$ such that
for every $g \in {\cal L}^1(c)$ of class $X$, $\lim_{n \rightarrow \infty} c(g-f_n)=0$.
\label{deforder}  
\end{definition}

\begin{lemma}
\begin{itemize}
\item let $X,Y \in L^1(c)$. If  $X \geq 0$ and $Y \geq 0$, then $X+Y \geq 0$. 
\item If there is in the class of $X$ a non negative function $f$ then $X \geq 0$
\item Let $X \in L^1(c)$; $|X| \in L^1(c)$. Furthermore $X \geq 0$ if and only if $X=|X|$ in $L^1(c)$.
\end{itemize}
\label{lemmaorder}
\end{lemma} 
\begin{proof} 
\begin{itemize}
\item The first part of the lemma is trivial. 
\item The second point follows from the inequality 
\begin{equation}
c(|f|-|f_n|) \leq c(f-f_n)
\label{eq0001}
\end{equation}
 Thus as $f=|f|$, $c(f-|f_n|) \leq c(f-f_n)$.
\item One can deduce from (\ref{eq0001}) that for all $X \in L^1(c)$, $|X| \in L^1(c)$. From point 2, $|X|-X \geq 0$.
  Thanks to  (\ref{eq0001}) and  the inequality  $c(|f|-f) \leq c(|f|-f_n)+c(f-f_n)$, it  follows that   $X \geq 0$ if and only if $X=|X|$ in $L^1(c)$. 
\end{itemize}
\end{proof}

\begin{proposition}
The relation $X \leq Y $ defined by $Y-X \geq 0$ defines a partial order on $L^1(c)$.
\label{order}
\end{proposition}
\begin{proof} 
1. Reflexivity is trivial: take  $f_n=0$ for all $ n$\\ 
2. Antisymmetry. Let $X \geq Y$ and $Y \geq X$. Let $h$ in the class of $X-Y$. By  definition there  are two sequences $f_n$ and $g_n$ of non negative functions in ${\cal L}$ such that $\lim_{n \rightarrow \infty}c(f_n-h)=0$ and  $\lim_{n \rightarrow \infty}c(g_n+h)=0$. It follows that $\lim_{n \rightarrow \infty}c(f_n+g_n)=0$. As $0 \leq |f_n-g_n| \leq f_n+g_n$, it follows that $\lim_{n \rightarrow \infty}c(|f_n-g_n|)=0$. However $\lim_{n \rightarrow \infty}c(f_n-g_n-2h)=0$. Thus $X-Y$, the class of $h$ is equal to $0$.\\
3. Transitivity follows from the first part of Lemma \ref{lemmaorder}.
\end{proof}

\subsection{Topological properties of the non negative part of the unit ball of $L^1(c)^*$}
\label{sectop}
For the definition of a Prokhorov capacity, see the Appendix (Section \ref{secap}).
\begin{proposition} 
 Let $c$ be a Prokhorov  capacity on  a metrizable and separable space $\Omega$. Every continuous linear form $L$ on $L^1(c)$ admits a representation:
\begin{equation}
L(f)= \int fd\mu \;\;\forall f \in L^1(c)
\label{eqreg10}
\end{equation}
where $\mu$ is a regular bounded signed  measure defined on a $\sigma$-algebra containing the Borel $\sigma$-algebra of $\Omega$.\\
If $L$ is a non negative linear form the  regular measure  $\mu$ is non negative finite.
\label{thmreg}
\end{proposition}
Following \cite{Bil} a bounded signed measure $\mu$ is called regular if for all Borel set $A$, for all $\epsilon>0$, there is a closed set $F$ and an open set $G$ such that $F \subset A \subset G$ and $|\mu|(G-F)<\epsilon$.\\
Notice that in \cite{FP}, the existence of a bounded measure $\mu$ satisfying equation (\ref{eqreg10}) is proved.  However the statement of  Proposition 11 of \cite{FP} does not give informations on the $\sigma$ algebra on which the measure $\mu$ is defined.  Therefore we have to go inside the proof.

\begin{proof}
\begin{itemize}
\item A metrizable space is completely regular and $c$ is a Prokhorov capacity so  Proposition 11 of \cite{FP} gives the existence of a measure $\mu$ satisfying equation (\ref {eqreg10}). We want now prove that $\mu$ is defined on the Borel $\sigma$ algebra. As in the proof of Proposition 11 of \cite{FP} let $Z$ be a compactification of $\Omega$, and $c'$ the capacity defined on $Z$ by $c'(g)=c(g_{|{\Omega}})$.  As $c$ is a Prokhorov capacity, from Proposition 11 of \cite{FP}, $c'(1_{Z-\Omega})=0$ and $L^1(c)=L^1(c')$.
\item
As $Z$ is a compact space, it follows from Theorem 3 of \cite{FP1} that every non negative linear form on $L^1(c')$ can be represented by a non negative measure  obtained from the Riesz representation theorem applied to ${\cal C}(Z)$. Therefore this measure  is defined on a $\sigma$-algebra  containing the Borel sets of $Z$. From Theorem 6 of \cite{FP1} every continuous linear form on $L^1(c)$ is the difference of two non negative linear forms, thus the bounded measure $\mu$ satisfying equation (\ref {eqreg10}) is defined on a Borel $\sigma$-algebra ${\cal B}$ containing the Borel $\sigma$-algebra of $Z$.
\item We want now prove that $\mu$ is defined on the Borel $\sigma$-algebra of $\Omega$.  $\mu$ is defined on the $\sigma$-algebra ${\cal F}$ obtained  by completion of ${\cal B}$ with the $\mu$-null sets. Notice  that from Theorem 3 of \cite{FP1},  every $c'$-negligible set (i.e. $c'(1_A)=0$) is also $\mu$-negligible. This is in particular the case for  $Z-\Omega$  which is therefore $\mu$-measurable. Every open set $V$ of $\Omega$ can be written $V=U\INTER \Omega$ for some open set $U$ of $Z$. Therefore $V$ belongs to ${\cal F}$. It follows that the measure $\mu$ defined on ${\cal F}$ is thus defined on the Borel $\sigma$-algebra of $\Omega$. As $\Omega$ is a metric space and $\mu$ is defined on the Borel $\sigma$-algebra  of $\Omega$, $\mu$ is regular from Theorem 1.1 of \cite{Bil}.
\end{itemize}
\end{proof}
Recall that the weak topology on  ${\cal M}_+(\Omega)$ the set of non negative finite measures on $(\Omega,{\cal B}(\Omega))$ is  the coarsest topology for which the mappings 
$$\mu \in {\cal M}_+(\Omega) \rightarrow \int fd\mu$$ are continuous for every given $f$ in ${\cal C}_b(\Omega)$.
\begin {proposition}
Let $c$ be a Prokhorov capacity on a metrizable separable space. The set of non negative linear forms on the Banach space  $L^1(c)$  is a subset of  ${\cal M}_+(\Omega)$. The weak* topology (i.e. the $\sigma(L^1(c)^*,L^1(c))$ topology) on the  non negative part $K_+$ of the unit ball  of  $L^1(c)^*$ coincides with  the restriction to $K_+$ of the weak topology  on  ${\cal M}_+(\Omega)$.
\label{lemmatop}
\end{proposition}

\begin{proof}
From Proposition \ref{thmreg}, every non negative linear form on $L^1(c)$ belongs to ${\cal M}_+(\Omega)$. 
Let  $\mu \in K_+$. As  ${{\cal C}_b(\Omega})$ is dense in the Banach space $L^1(c)$, the open sets
$$V_{f_1,f_2,...f_n,\epsilon}(\mu)=\{\nu \in   K_+ \;\; |\forall i \in \{1,...n\}, \;|\mu(f_i)-\nu(f_i)|<\epsilon\}$$
with  $f_i \in {\cal C}_b(\Omega)$ form a basis of neighborhoods of  $\mu$ in $K_+$ for the weak* topology. Thus the weak* topology on $K_+$ coincides with the weak topology. \end{proof}
\begin{proposition}
Let $c$ be a Prokhorov capacity on a metrizable separable space $\Omega$. The set $K_+$   is compact metrizable for the weak* topology (i.e. the $\sigma(L^1(c)^*,L^1(c))$ topology), as well as for the weak topology. 
\label{propcm}
\end{proposition}
\begin{proof} 
Prove first that $K_+$ is metrizable for the weak* topology.
From Proposition \ref{lemmatop}, the weak* topology on $K_+$ coincides with the restriction to $K_+$ of  the weak topology on $ {\cal M}_+(\Omega)$. As  $\Omega$ is metrizable and separable,   $ {\cal M}_+(\Omega)$ is also metrizable and separable for the weak topology from \cite{B} Section 5. Thus  $K_+$ is metrizable for the weak* topology.\\
From  Banach Alaoglu Theorem,  (theorem V 4 2 of \cite{DS}) the closed unit  ball of the dual space of a Banach space is always compact for the weak* topology. As $K_+$ is a  closed subset of  this unit ball for the weak* topology, it is also compact. This proves the result for the weak* topology. From Proposition \ref{lemmatop}, $K_+$ is also metrizable compact for the weak topology. 
\end{proof}
\begin{corollary}
 Assume that ${\Omega}$ is a  Polish space. For every capacity $c$ on $\Omega$, the set $K_+$     is compact metrizable for the weak* topology.
\label{corollaryPolish}
\end{corollary}
\begin{proof} From \cite{FP}, see also the Appendix (Section \ref{secap}), every capacity on a Polish space is a Prokhorov capacity, and thus the result follows from Proposition \ref{propcm}. 
\end{proof}
In the particular case of a compact metrizable space, we obtain the following stronger result.
\begin{proposition}
Let $\Omega$ be a metrizable compact space. Let $c$ be a capacity on $\Omega$. Then the Banach space $L^1(c)$ is separable and the unit ball of $L^1(c)^*$ is metrizable compact for the weak* topology.
\label{propcomp}
\end{proposition}
\begin{proof} As $\Omega$ is a metrizable compact space, ${\cal C}(\Omega)$ is separable from Thm 1  Section 3 of \cite{BT}.  Thus for every capacity $c$ on $\Omega$, $L^1(c)$ is also separable. Then  from Theorem V 5 1 of \cite{DS}, the unit ball of $L^1(c)^*$ (and not only its non negative part) is metrizable  compact for the weak* topology.
\end{proof}

\section{Representation of a convex risk measure on $L^1(c)$}
\label{L1}
In this section, $c$ denotes a  Prokhorov capacity on a metrizable separable space $\Omega$.
Recall that a partial order has been defined on $L^1(c)$ in Section \ref{sec1.1}. We can define convex risk measures in the usual way as follows.
\begin{definition}
Let $\rho:L^1(c) \rightarrow \R$.
\begin{itemize}
\item $\rho$ is monotonic if 
$\rho(X) \geq \rho(Y)$ for every $X, Y \in L^1(c)$, such that $X \leq Y$.
\item  $\rho$ is convex if for every $X, Y \in L^1(c)$, for every $0 \leq \lambda \leq 1$, $\rho( \lambda X + (1-\lambda)Y \leq \lambda \rho(X) + (1-\lambda)\rho(Y)$
\item $\rho$ is translation invariant if  $\rho(X+a)=\rho(X)-a$ for every $X \in L^1(c)$ and $a \in \R$.
\end{itemize}
$\rho$ is a convex risk measure if it satisfies all these conditions.
\label{defmes}
\end{definition}

\subsection{Representation for convex risk measures}
Duality results for risk measures are well known in other settings.  A duality result was first proved in the case of risk measures on $L^{\infty}$ spaces assuming furthermore continuity from below. Duality results are based on the Fenchel Legendre duality, generalized to the  context of locally convex topological spaces by Rockafellar \cite{Ro}. This is the generalized version that we need here.  No additional hypothesis is needed in order to prove the dual representation result. The important and new discussion will be developed in  Subsection \ref{subsecwc} using the topological results proved in Section \ref{sectop}.
\begin{theorem}
Let $\rho$ be a convex risk measure on $L^1(c)$. Then, $\rho$ is continuous and admits a representation of the form:
\begin{equation}
\forall X\in L^1(c),\ \rho\left( X\right)=\sup_{Q\in\mathcal{P'}}(
E_{Q}[-X]-\alpha\left( Q\right))
\label{eqmag}
\end{equation}
where \begin{equation}
\alpha\left(Q\right)=\sup_{X\in L^{1}(c)}\left( E_{Q}[-X]-\rho\left( X\right)\right)
\label{eqma}
\end{equation}
 $\mathcal{P'}$ is the set of probability measures on $(\Omega,{\cal B}(\Omega))$ belonging to $L^1(c)^*$.
\label{thmrep1}
\end{theorem}
\begin{proof}
The continuity of $\rho$ 
follows from Theorem 1 of \cite{BF}. \\
We call $\alpha$ the function on $L^1(c)^{\star}$  defined by:
$$\forall \mu\in L^1(c)^{\star},\ \alpha\left( \mu\right)=\sup_{X\in L^1(c)}\left( \mu\left(X\right)-\rho\left( X\right)\right)$$
 As the dual of $L^1(c)^*$ (with the weak * topology) is $L^1(c)$, the locally convex topological spaces $L^1(c)$ and $L^1(c)^*$ are paired in the sense of \cite{Ro}. 
$\rho$ is continuous,  we can thus apply Theorem 5 in Rockafellar \cite{Ro}. We get the following equality:
$$\forall X\in L^1(c),\ \rho\left( X\right)= \sup_{\mu\in L^1(c)^{\star}}\left( \mu\left( X\right)-\alpha\left( \mu\right)\right)$$
In the supremum above, we 
can obviously restrict to the elements $\mu$ of $L^1(c)^{\star}$ such that $\alpha \left( \mu\right)<+\infty$.\\
Let $\mu_0\in L^1(c)^{\star}$ such that $\alpha\left( \mu_0\right)<+\infty$, we first  prove that $-\mu_0$ is a positive linear form. Let $X\in L^1(c)$ such that 
$X\geq 0$. For all $\lambda>0$, using  the monotonicity of $\rho$,  $\rho\left( \lambda X\right)\leq \rho\left( 0\right)$, which implies  that
$$\lambda \mu_0\left( X\right)-\alpha\left( \mu_0\right) \leq \rho\left(0\right)$$
$\rho\left( 0\right)$ and $\alpha(\mu_0)$ are finite  and the above inequality  is satisfied for all $\lambda>0$, thus $\mu_0\left( X\right)\leq 0$.\\
From Proposition \ref{thmreg}, $-\mu_0$ is represented by a  finite non negative measure defined on $(\Omega,{\cal B}(\Omega))$.
Thanks to  the translation invariance of $\rho$, for all $\lambda \in\R$, $\rho\left( \lambda\right)=\rho\left( 0\right)-\lambda$, which means that:
$$\rho\left( 0\right)=\lambda+ \sup_{\mu\in L^1(c)^{\star}}\left(\lambda \mu\left( 1\right)-\alpha\left( \mu\right)\right)\geq \lambda\left(1+ \mu_0\left( 1\right)\right)-\alpha\left( \mu_0\right)$$
We conclude as  above  that $1+ \mu_0\left( 1\right)=0$. Thus, $-\mu_0$ is a probability measure on $(\Omega,{\cal B}(\Omega))$ and $-\mu_0 \in L^1(c)^*$.
\end{proof} 

\subsection{Risk measures represented by a weakly relatively compact set of probability measures}
\label{subsecwc}
 In this section we want to characterize risk measures $\rho$ on $L^1(c)$ admitting a dual representation with a  relatively  compact set of probability measures for the weak* topology. 
 \begin{definition} A convex risk measure $\rho$ on $L^1(c)$ is normalized if $\rho(0)=0$.
\end{definition}
\begin{proposition}
Let $\rho:L^1(c) \rightarrow \R$ be a normalized convex risk measure.
The following conditions are equivalent:
\begin{enumerate}
\item $\rho$  is majorized by a sublinear risk measure
\item $\forall X \in L^1(c)$, $\sup_{\lambda >0}\frac{\rho(\lambda X)}{\lambda} < \infty$
\item there exits $K>0$ such that $\forall X \in L^1(c)$, $|\rho(X)| \leq K c(X)$
\item $\rho$ is represented by a   set ${\cal Q}$ of probability measures in ${L^1}(c)^*$ relatively compact for the weak* topology, i.e. 
\begin{equation}
\forall X\in L^1(c),\ \rho\left( X\right)=\sup_{Q\in {\cal Q}}(E_{Q}[-X]-\alpha\left( Q\right))
\label{eqmag2}
\end{equation}
\end{enumerate}
\label{prop6}
\end{proposition}
Before giving the proof of the Proposition, we prove the following Lemma
\begin{lemma} Let ${\cal Q}$ be  a set of probability measures on $(\Omega, {\cal B}(\Omega))$ such that ${\cal Q} \subset L^1(c)^*$. Assume that ${\cal Q}$ is relatively compact for the weak* topology $\sigma((L^1(c)^*, L^1(c))$. Then ${\cal Q}$ is contained in some closed ball of $ L^1(c)^*$ and the weak* closure of ${\cal Q}$ is also compact for the weak topology.
\label{lemma4.1}
\end{lemma}
\begin{proof}
Denote $\overline {\cal Q}$ the closure of ${\cal Q}$ for the weak* topology. $\overline {\cal Q}$ is compact. Let $X \in L^1(c)$. The map $Q \rightarrow E_Q(X)$ is continuous for the weak* topology, thus $\sup_{Q \in \overline {\cal Q}}|E_Q(X)| < \infty$. From Banach Steinhauss Theorem (cf \cite{Ru}), it follows that $\overline {\cal Q}$ is contained in some closed ball of $L^1(c)^*$, and thus in the non negative part of this closed ball. From Proposition \ref{lemmatop}, $\overline {\cal Q}$ is weakly compact. 
\end{proof}
We can now give the proof of Proposition \ref{prop6}.
\begin{proof}
Consider the dual representation of $\rho$ given by equation (\ref{eqmag}). Denote 
${\cal Q}=\{Q \in {\cal P}'\;|\; \alpha(Q)< \infty\}$. Then
\begin{equation}
\forall X\in L^1(c),\ \rho\left( X\right)=\sup_{Q\in {\cal Q}}(E_{Q}(-X)-\alpha\left( Q\right))
\label{eqmag2bis}
\end{equation}
{\it1}.  implies {\it 2}. Let $\rho_1$ be a sublinear risk measure majorizing $\rho$.
Then for every $\lambda  \in \R^+_{*}$, $\rho(\lambda X) \leq \lambda \rho_1(X)$.
Thus $\sup_{\lambda >0}\frac{\rho(\lambda X)}{\lambda} \leq \rho_1(X)$, and {\it 2} is proved.\\
{\it2}.  implies {\it 3}. 
For every $X \in L^1(c)$, denote $\beta_X=\sup_{\lambda >0}\frac{\rho(\lambda X)}{\lambda}$. From the dual representation (\ref{eqmag2bis}), applied with $\lambda X$ for every $\lambda >0$, it follows that $\forall Q \in {\cal Q}$, $E_Q(-X) \leq \beta_X$, and thus $\sup_{Q \in {\cal Q}}E_Q(-X) \leq \beta_X < \infty$ for every $X \in L^1(c)$. With $X =-|Y|$, we get that 
\begin{equation}
\forall Y \in L^1(c), \;\sup_{Q \in {\cal Q}}|E_Q(Y)|<\infty 
\label{eqBS}
\end{equation}
    $L^1(c)$ is a Banach space and from Theorem \ref{thmrep1}, every $E_Q$ is a continuous linear form on $L^1(c)$. Denote $||E_Q||$ its norm.
From  Banach Steinhauss  Theorem,  equation (\ref{eqBS}) implies the existence of  $K>0$ such that $\sup_{Q \in {\cal Q}} ||E_Q|| \leq K$. Notice that from the normalization condition ($\rho(0)=0$) it follows from equation (\ref{eqma}) that for every $Q$, $\alpha(Q) \geq 0$. Thus from the representation (\ref{eqmag2bis}), for every $X \in L^1(c)$, 
\begin{equation}
\rho(X) \leq K c(X)
 \label{eqB}
\end{equation}
From the  convexity, the  monotonicity of $\rho$ and  $\rho(0)=0$, it follows that \begin{equation}
-\rho(X) \leq \rho(-X) \leq \rho(-|X|) \leq K c(-|X|)=K c(X)
 \label{eqB2}
\end{equation} 
Thus from equations (\ref{eqB}) and (\ref{eqB2}), for every $X \in L^1(c)$, 
$$|\rho(X)| \leq K c(X)$$
This proves {\it 3}.\\
 {\it 3}. implies {\it 4}.
 From the representation of $\rho$, equation (\ref{eqmag2bis}) applied with $- \lambda |X|$ for every $\lambda>0$, it follows from hypothesis {\it 3}.  that for every $Q \in {\cal Q}$ $||E_Q|| \leq K$.
This means that ${\cal Q}$ is contained in a closed ball of the dual of $L^1(c)$. Every such closed ball  is compact for the weak* topology (Banach Alaoglu Theorem). Thus ${\cal Q}$ is relatively compact for the weak* topology.\\
{\it 4}. implies {\it 1}.
 $\rho$ is represented by a  set of probability measures  ${\cal Q} \subset L^1(c)^*$ relatively compact for the weak * topology. From Lemma \ref{lemma4.1}, ${\cal Q}$ is contained in some closed ball of $L^1(c)^*$. Define $\rho_1$  by $\rho_1(X)= \sup_{Q \in {\cal Q}}E_Q(-X)$. As ${\cal Q}$ is bounded, $\rho_1(X)$ is finite for every $X$ in $L^1(c)$. It is easy to verify that $\rho_1$ is a sublinear risk measure and that $\rho$ is majorized by $\rho_1$. 
\end{proof}

\begin{theorem}
Let ${\rho}$ be a convex risk measure on $L^1(c)$. Assume that $\rho$ is  represented by 
$$\rho(X)= \sup_{Q \in {\cal Q}} (E_Q(-X)-\alpha(Q))$$
where ${\cal Q}$ is a set of probability measures in $L^1(c)^*$  relatively compact  for the  weak* topology. Let $\overline{\cal Q}$ be the closure of ${\cal Q}$ for the weak* topology.
Then 
\begin{itemize}
\item
$\overline {\cal Q}$ is 
metrizable compact both   for the  weak* topology and the weak topology.
\item
 For every $X \in L^1(c)$, there is a probability measure $Q_X \in \overline {\cal Q}$ such that 
\begin{equation}
\rho(X)=  E_{Q_X}(-X)-\alpha(Q_X)
\label{eqatt}
\end{equation}
\end{itemize}
\label{thmmax}
\end{theorem}
\begin{proof}
\begin{itemize}
\item
From Lemma \ref{lemma4.1}, $\overline {\cal Q}$ is contained in a closed ball of $L^1(c)^*$ and is compact both for the weak and the weak* topology.  From Proposition \ref{propcm} it is metrizable compact.
\item Let $X \in L^1(c)$.
Let $Q_n$ be a sequence of elements in ${\cal Q}$ such that  for every $n$, 
\begin{equation}
\rho(X) - \frac{1}{n} < E_{Q_n}(-X) - \alpha(Q_n) \leq \rho(X)
\label{eq1.0}
\end{equation}
As $\overline {\cal Q}$ is metrizable compact for the weak* topology, there is a subsequence $Q_{\phi(n)}$ converging to $\tilde Q \in \overline{\cal Q}$, satisfying the inequality 
\begin{equation}
E_{\tilde Q}(-X)-\frac{1}{n} < E_{Q_{\phi(n)}}(-X) < E_{\tilde Q}(-X)+ \frac{1}{n}
\label{eq1.1}
\end{equation} 
From inequality (\ref{eq1.0}) applied with $Q_{\phi(n)}$, inequality  (\ref{eq1.1}) and the inequality $\phi(n) \geq n$, it follows that 

\begin{equation}
E_{\tilde Q} (-X)-\rho(X)- \frac{1}{n} < \alpha(Q_{\phi(n)})
 < E_{\tilde Q} (-X)-\rho(X)+ \frac{2}{n}
\label{eq1.2}
\end{equation}
Let $Y \in L^1(c)$.  Let $\epsilon>0$. There is $N(Y)$ such that for every $n > N(Y)$, $E_{\tilde Q}(-Y) < E_{Q_{\phi(n)}}(-Y)+\epsilon$. $N(Y)$ can be chosen such that $N(Y) \geq \frac{1}{\epsilon}$
Then  for $n \geq N(Y)$,

\begin{eqnarray}
E_{\tilde Q}(-Y)- \rho(Y)  & \leq &  \alpha(Q_{\phi(n)})+ \epsilon \nonumber \\
                 &  \leq & E_{\tilde Q}(-X)- \rho(X) + \frac{2}{n} + \epsilon 
\nonumber \\
                 &  \leq &  E_{\tilde Q}(-X)- \rho(X) + 3 \epsilon 
\label{eq1.4}
\end{eqnarray}
 As the inequality is satisfied for every $Y$ and every $\epsilon>0$, it follows that 
$$\alpha(\tilde Q)= \sup_{Y \in L^1(c)} (E_{\tilde Q}(-Y)- \rho(Y)) \leq E_{\tilde Q}(-X)- \rho(X)$$
And thus $$\rho(X)=  E_{\tilde Q}(-X)-\alpha(\tilde Q)$$
\end{itemize}
\end{proof} 
\begin{proposition}
Let $\rho$ be a normalized convex risk measure  on $L^1(c)$ majorized by a sublinear risk measure. There is  a countable set $\{R_n, \; n \in \N\}$ of probability measures  belonging to $L^1(c)^*$, which   is  relatively compact for the weak* topology of $L^1(c)^*$ and also for the weak topology and such that 
\begin{equation}
\forall X\in L^1(c),\ \rho\left( X\right)=\sup_{n \in \N}(E_{R_n}[-X]-\alpha( R_n))
\label{eqmag3}
\end{equation}
where 
\begin{equation}
\alpha\left(R\right)=\sup_{X\in L^{1}(c)}\left( E_{R}[-X]-\rho\left( X\right)\right)
\label{eqma2}
\end{equation}
\label{proprepnum}
\end{proposition}
\begin{proof}
From Proposition \ref{prop6}, there is a set ${\cal Q}$ of probability measures in $L^1(c)^*$, relatively compact for the weak* topology such that equation (\ref{eqmag2}) is satisfied. From Lemma \ref{lemma4.1}, ${\cal Q}$ is contained in $mK_+$,  the non negative part of a certain  closed ball  of $L^1(c)^*$.
 From Proposition \ref{propcomp},  $mK_+$,   is metrizable compact for the weak* topology. There is thus a countable dense set $(Q_n)_{n \in \N}$ in $mK_+$. Denote $d$ a distance on $mK_+$ defining the weak* topology. For every $Q \in mK_+$, let $\overline {B(Q,r)}= \{R \in mK_+\; |d(Q,R) \leq r\}$. The set $\overline {B(Q,r)}$ is compact for the weak* topology. The penalty  $\alpha$ defined on $L^1(c)^*$ by equation \ref{eqma2} is l.s.c. thus for every $n \in \N$ and $k \in \N^*$ there is $R_{n,k}$ in $\overline {B(Q_n,\frac{1}{2^k})}$ such that $\alpha(R_{n,k})= \min\{\alpha(Q),\; Q \in
\overline {B(Q_n,\frac{1}{2^k})}\}$.\\
Let $X \in L^1(c)$.  From Theorem  \ref{thmmax}, there is $Q_X \in {\cal Q}$ such that 
$\rho(X)=E_{Q_X}(-X) -\alpha(Q_X)$. For all  $ \epsilon>0$,  there is $\eta>0$ such that $\forall Q \in  {B(Q_X,\eta)}$, $|E_{Q_X}(-X)-E_Q(-X)| < \epsilon$. Let $k$ such that $\frac{1}{2^{k-1}}<\eta$. Let $n$ such that $Q_X \in B(Q_n,\frac{1}{2^k})$
then $E_{R_{n,k}}(-X)-\alpha(R_{n,k})>\rho(X)-\epsilon$.
It follows that $\{R_{n,k},\; n \in \N,\; k \in \N^*\}$ is a countable set weakly relatively compact (as it is contained in $mK_+$)  satisfying the required condition. 
\end{proof}

\begin{theorem} 
 Every convex risk measure on $L^1(c)$ can be represented by a countable set of probability measures $\{R_n, \; n \in \N\}$ belonging  to $L^1(c)^*$.
\begin{equation}
\forall X\in L^1(c),\ \rho\left( X\right)=\sup_{n \in \N}(E_{R_n}(-X)-\alpha( R_n))
\label{eqmag4}
\end{equation}
where $\alpha(R)$ is given by equation (\ref{eqma2}).
\label{thmrepnum}
\end{theorem}

\begin{proof}  From Theorem \ref{thmrep1}, $\rho$ has a dual representation given by equation (\ref{eqmag}). Denote then $\rho_m(X)= \sup_{Q \in mK_+}(E_Q(-X)-\alpha(Q))$. Even if $\rho_m$ is not necessarily normalized, all the arguments of the proof of Proposition \ref{proprepnum} apply as $mK_+$ is metrizable compact for the weak* topology and $\alpha$ is l.s.c..  Thus $\rho_m$ has a representation with a countable set of probability measures. As $\rho= \sup_{m \in \N}\rho_m$, this gives the result.
\end{proof}

\section{Equivalence class of probability measures associated to a non dominated set of probability measures}
\label{sec2}

Let $\Omega$ be a  metrizable and separable space. In this section we study a capacity defined from  a weakly relatively compact set of probability measures ${\cal P}$ possibly non dominated.

\begin{definition}
Let  ${\cal P}$ be a weakly relatively compact set of probability measures on $(\Omega,{\cal B}(\Omega))$. Let $1 \leq p < \infty$. The capacity $c_{p,{\cal P}}$ is defined on ${\cal C}_b(\Omega)$ by 
\begin{equation}
c_{p,{\cal P}}(f)=\sup_{P \in {\cal P}} E_P(|f|^p)^{\frac{1}{p}}
\label{eqcap0}
\end{equation} 
and extended to every function on $\Omega$ as explained in Section \ref{sec1.1}, equations (\ref{eqsci}) and (\ref{eqsci2}).

\label{defcap}
\end{definition}
Notice that as ${\cal P}$ is a weakly relatively compact set of probability measures,  $c_{p,{\cal P}}$ is a capacity (see Proposition I.3 of \cite{K} or the Appendix, Section \ref{secap}). The Banach space associated to the capacity $c_{p,{\cal P}}$ is denoted $L^1(c_{p,{\cal P}})$.  When there is no ambiguity on the set ${\cal P}$ we simply write $c_p$ for $c_{p,{\cal P}}$.\\
When ${\cal P}=\{\mu_0\}$, $L^1(c_{p,\{\mu_0\}})= L^1(\Omega,{\cal B}(\Omega),\mu_0)$.  A non negative   measure $\mu$ on $(\Omega,{\cal B}(\Omega))$ belongs to the (usual)  equivalence class of  the probability measure $\mu_0$ if and only if
$\forall A \in{\cal B}(\Omega), \;\; \mu(A)=0 \;\Longleftrightarrow \; \mu_0(A)=0$\\
 Equivalently, for $\mu$ in the dual of  $L^1(\Omega,{\cal B}(\Omega),\mu_0)$,
$$\mu \sim \mu_0  \;\Longleftrightarrow \; [\forall X \in L^{1}(\Omega, {\cal B}(\Omega),\mu_0)_+,\;\;  X=0 \;\Longleftrightarrow \;\int X d \mu=0 ]$$
We address the following question: When ${\cal P}$ is weakly relatively compact can one associate a probability measure $P$ to $L^1(c_{p,{\cal P}})$ characterizing the  null elements in the cone $L^1(c_{p,{\cal P}})_+$, i.e. such that $\forall X \in L^1(c_{p,{\cal P}})_+,\;\;X=0 \;\Longleftrightarrow \;E_P(X)=\int X d P=0$ ? If yes, can one define a natural equivalence relation so that one gets  a unique equivalence class of such probability measures?
Notice that when  ${\cal P}$ is not finite,  characteristic functions of  Borelian sets are not all in  $L^1(c_{p,{\cal P}})$. 

\subsection{Properties of the capacity}

\begin{lemma}
For all $X$ in $L^1(c_{p,{\cal P}})$, $c_{p,{\cal P}}(X)= \sup_{Q \in {\cal P}}E_Q(|X|^p)^{\frac{1}{p}}$.
\label{lemmacap}
\end{lemma}
\begin{proof}
Denote $c_p=c_{p,{\cal P}}$. For all $f,g$ in  ${\cal C}_b(\Omega)$, for all $Q \in {\cal P}$,
$$|E_Q(|f|^p)^{\frac{1}{p}}-E_Q(|g|^p)^{\frac{1}{p}}| \leq E_Q(|f-g|^p)^{\frac{1}{p}} \leq c_p(|f-g|)$$
As ${\cal C}_b(\Omega)$ is dense in $L^1(c_p)$ for the $c_p$ norm it follows that 
 for every $X \in L^1(c_p)$, $g \in {\cal C}_b(\Omega)$, and $Q \in {\cal P}$,
\begin{equation}
|E_Q(|X|^p)^{\frac{1}{p}}-E_Q(|g|^p)^{\frac{1}{p}}| \leq c_p(|X-g|)
\label{eqcapX}
\end{equation}
From (\ref{eqcapX}) it follows that 
\begin{equation}
E_Q(|X|^p)^{\frac{1}{p}} \leq c_p(X) \;\forall Q \in {\cal P}
\label{eqcapX1}
\end{equation}
For every $X \in L^1(c_p)$, for every $\epsilon>0$ there is $g \in {\cal C}_b(\Omega)$ such that 
\begin{equation}
c_p(X-g) \leq \epsilon
\label{eq0007}
\end{equation} 
 From Definition \ref{eqcap0}, there is $Q_0 \in {\cal P}$ such that 
\begin{equation}
c_p(g) \leq E_{Q_0}(|g|^p)^{\frac{1}{p}}  + \epsilon
\label{eq0008}
\end{equation} 
 As $c_p(X) \leq c_p(g) +\epsilon$ it follows from equations (\ref{eqcapX}) (\ref{eq0007}) and(\ref{eq0008}) that $c_p(X) \leq  \sup_{Q \in {\cal P}}E_Q(|X|^p)^{\frac{1}{p}}$.
The result follows from (\ref{eqcapX1}).
\end{proof} 

\begin{theorem}
Assume that $\Omega$ is a  Polish space.  There is a countable subset ${\cal Q}$ of ${\cal P}$,  ${\cal Q}=\{P_n, \;n \in \N\}$,  such that for every $X \in { L}^1(c_{p,{\cal P}})$, for every $p \in [1,\infty[$,
\begin{equation}
c_{p,{\cal P}}(X)=\sup_{n \in \N} (E_{P_n}(|X|^p))^{\frac{1}{p}}
\label{eqnum0}
\end{equation} The capacities $c_{p,{\cal P}}$ and $c_{p,{\cal Q}}$ defined on ${\cal C}_b(\Omega)$ by equation  (\ref{eqcap0}) and extended to real functions using formulas (\ref{eqsci}) and (\ref{eqsci2}) are equal. The associated Banach spaces   are equal: $L^1(c_{p,{\cal P}})= L^1(c_{p,{\cal Q}})$. 
\label{Thm1}
\end{theorem}
\begin{proof}
From the previous Lemma, applied with $p=1$,  it follows that  the set ${\cal P}$ is contained in ${K}_+$, the non negative part of the unit ball of the dual of $L^1(c_{1,{\cal P}})$. $\Omega$ is a Polish space, so from Corollary \ref{corollaryPolish}, ${K}_+$ is metrizable compact  for the   weak* topology. Thus $\overline {\cal P}$, the closure of ${\cal P}$ for the weak* topology, is metrizable compact.  There is then in $ {\cal P}$ a countable set $(P_n)_{n \in \N}$ dense in $\overline {\cal P}$ for the weak* 
topology.
It follows that for every $X \in L^1(c_{1,{\cal P}})$, $\sup_{Q \in {\cal P}}E_Q(|X|)=\sup_{n \in \N}E_{P_n}(|X|)$. The equation (\ref{eqnum0}) follows for every $p \geq 1$ for every $X \in {\cal C}_b(\Omega)$.  \\
The two capacities  $c_{p, {\cal P}}(f)=\sup_{P \in {\cal P}}E_P(|f|^p)^{\frac{1}{p}}$   and $c_{p, {\cal Q}}=\sup_{Q \in {\cal Q}}E_Q(|f|^p)^{\frac{1}{p}}$ coincide on ${\cal C}_b(\Omega)$. By  definition of the  extension of a capacity to the set of all functions on $\Omega$, these extensions are the same. Therefore $L^1(c_{p,{\cal P}})= L^1(c_{p,{\cal Q}})$.
\end{proof} 
In the following proposition we study possible  extensions of the equation (\ref{eqcap0}).
$$ $$
\begin{proposition}
Let $c_p=c_{p,{\cal P}}$.
\begin{itemize}
\item  For every non negative bounded lower semi-continuous map $g$,  
\begin{equation}
c_p(g)= \sup_{Q \in {\cal P}}E_Q(g^p)^{\frac{1}{p}}
\label{eqop}
\end{equation}
\item For every Borelian map $f$, 
\begin{equation}
\sup_{Q \in {\cal P}}E_Q(|f|^p)^{\frac{1}{p}} \leq c_p(f)
\label{eqbor}
\end{equation}
\end{itemize} 
\label{propmes}
\end{proposition}
\begin{proof}
\begin{itemize}
\item The proof of the first part of 
Proposition \ref{propopenclosed} which was given for the characteristic function of an open set applies without any change to every  non negative bounded l.s.c. function $g$. Thus there is an increasing sequence of continuous functions 
$h_n$ with limit $g$ and such that $c_p(g)=\lim c_p(h_n)$. As $g$ is bounded, $c_p(g)$ is finite. Let $\epsilon>0$. There is $n$ such that $c_p(g)-\epsilon \leq c_p(h_n) \leq c_p(g)$. 
By definition of $c_p$ on ${\cal C}_b(\Omega)$, there is $Q_n$  in ${\cal P}$ such that $c_p(h_n)-\epsilon  \leq E_{Q_n}(h_n^p)^{\frac{1}{p}} \leq c_p(h_n)$. Thus 
\begin{equation}
E_{Q_n}(g^p)^{\frac{1}{p}} \geq c_p(g)-2 \epsilon
\label{equsci}
\end{equation}
 On the other hand for all $Q$ in ${\cal P}$, $E_Q(h_n^p)^{\frac{1}{p}} \leq c_p(h_n) \leq c_p(g)$.  From the monotone convergence theorem it follows that 
\begin{equation}
\forall Q \in {\cal P}, \;
E_Q(g^p)^{\frac{1}{p}} \leq c_p(g)
\label{equsci2}
\end{equation}
Thus from equations (\ref{equsci}) and (\ref{equsci2}) we get that 
\begin{equation}
c_p(g)= \sup_{Q \in {\cal P}}E_Q(g^p)^{\frac{1}{p}}
\label{equsci3}
\end{equation}
\item  Let $f$ be a Borelian map.  If $c_p(f)= +\infty$, the result is trivial. Assume that $c_p(f) <\infty$.  
 Let $\epsilon>0$. By definition of  $c_p(f)$, (equation \ref{eqsci2}), there is $g$ l.s.c., $g \geq |f|$ such that $c_p(g)<c_p(f)+\epsilon$. As $g$ is l.s.c., we already know that  
 $\sup_{Q \in {\cal P}} E_Q(|g|^p)^{\frac{1}{p}} = c_p(g)$. As  $f$ is Borel measurable, for all $Q \in {\cal P}$, $E_Q(|f|^p)^{\frac{1}{p}}$ is defined. As $g \geq |f|$ it follows that 
 $E_Q(|f|^p)^{\frac{1}{p}} \leq c_p(f)+\epsilon$.
This inequality is true for every $\epsilon$ and every $Q \in {\cal P}$. This proves the announced result for every $f$ Borel measurable.
 \end{itemize}
\end{proof} 
\begin {remark} 
\begin{itemize}
\item For every open subset $V$ of $\Omega$, $1_V$ is lower semi-continuous, so from Proposition \ref{propmes}, $c_p(1_V)=\sup_{Q \in {\cal P}}Q(V)^{\frac{1}{p}}$. 
\item However there are  Borelian subsets of $\Omega$ for which the equality 
$c_p(1_A)=\sup_{Q \in {\cal P}}Q(A)^{\frac{1}{p}}$ is not satisfied.\\
For example let $\Omega=[0,1]$. Let $x_n \in]0,1[$ be a sequence converging to $0$. Let $A=[0,1]-\{x_n,n \in \N\}$. Let $Q_n = \delta_{x_n}$. Let ${\cal P}=\{Q_n,\; n \in \N\}$. ${\cal P}$ is weakly relatively compact. Let  $f$ l.s.c. such that $1_A \leq f \leq 1$. For every $\eta>0$,   $V=\{x |f(x)> 1 - \eta \}$ is an open set containing $A$. As $0 \in A$, there is $\epsilon>0$ such that $[0,\epsilon[ \subset V$. So there is $N \in \N$ such that $x_n \in V\; \forall n \geq N$. So $E_{Q_n}(f^p)=(f(x_n))^p>(1 - \eta)^p $. From equation (\ref{eqop}),   $ 1 \geq c_p(f)=\sup_{n \in \N}(E_{Q_n}(f^p))^{\frac{1}{p}}>1-\eta$ for every $\eta>0$.  Thus
  $c_p(f)=1$. It follows that  $c_p(1_A) = 1$. On the other hand $Q_n(1_A)=0$ for all $n \in \N$. Therefore $\sup_{Q \in {\cal P}}Q(A)^{\frac{1}{p}}=0$. This gives a counterexample.
\end{itemize}
\end{remark}
\subsection{Canonical equivalence  class of non negative measures associated to $c_{p}$}
In all this section, we assume that $\Omega$ is a Polish space. We denote $c_p$ the capacity defined on ${\cal C}_b(\Omega)$ by $c_p(f)=\sup_{Q \in {\cal P}}E_Q(|f|^p)^{\frac{1}{p}}$.

\begin{definition}
 ${\cal M}^+(c_p)$ is the set of non negative finite measures on $(\Omega,{\cal B}(\Omega))$ defining an element of $L^1(c_p)^*$.
\label{defflp}
\end{definition}
 In the following we identify an element $\mu$ of ${\cal M}^+(c_p)$ with its associated linear form on $L^1(c_p)$.
 
\begin{remark}
A non negative  finite measure $\mu$ on $(\Omega,{\cal B}(\Omega))$ belongs to ${\cal M}^+(c_p)$ if and only if there is a constant $K>0$ such that $\forall f \in {\cal C}_b(\Omega), |\mu(f)| \leq K c_p(f)$. It follows easily that every element in the weak closure of the convex hull of ${\cal P}$ defines an element of ${\cal M}^+(c_p)$.
\end{remark}
\begin{definition}
Define on ${\cal M}^+(c_p)$ the relation ${\cal R}_{c_p}$ by 
\begin{equation}
\mu {\cal R}_{c_p} \nu\; \; \Longleftrightarrow
\label{eqeq}
\end{equation}
\begin{equation*} 
\{X \in L^1(c_p) , X \geq 0 \;|\;\mu(X)=0\}= \{X \in L^1(c_p), X \geq 0 \;|\;\nu(X)=0\}
\end{equation*}
\label{defreleq}
\end{definition}
The following lemma is trivial
\begin{lemma}
${\cal R}_{c_p}$ defines an equivalence relation on ${\cal M}^+(c_p)$.
\label{lemreleq}
\end{lemma}
\begin{definition}
Let $\mu \in {\cal M}^+(c_p)$. The $c_p$-class of $\mu$ is the equivalence class of $\mu$ for the equivalence relation $ {\cal R}_{c_p}$.
\label{defclass}
\end{definition}
\begin{theorem} 
 To every  weakly relatively compact set ${\cal P}$ of probability measures  on $(\Omega,{\cal B}(\Omega))$, possibly non dominated,  can be associated canonically  a  $c_p$-class of non negative measures on $(\Omega,{\cal B}(\Omega))$ such that an element $\mu$ of ${\cal M}^+(c_p)$ belongs to this class if and only if 
$$\forall X \in L^1(c_p), X \geq 0, \;\;\;\;\;\; \;\{\mu(X)=0\} \;\Longleftrightarrow\;\{ X=0\; \text{in}\; L^1(c_p)\}$$ This class is referred to as the canonical $c_p$-class.\\
 For every set $\{Q_n, \; n \in \N\}$ of probability measures  on $(\Omega,{\cal B}(\Omega))$ such that the equality 
(\ref{eqnum0}) is satisfied for all  $X \in L^1(c_p)$, for $\alpha_n >0 $ such that $\sum_{n \in \N} \alpha_n=1$  the probability measure $\sum_{n \in \N} \alpha_n Q_n$ belongs to the canonical  $c_p$-class.
\label{proppolar}
\end{theorem}
\begin{proof}
Let $p \in [1, \infty[$. Let
$\{Q_n\}$ be a countable set of probability measures  such that the equality (\ref{eqnum0}) is satisfied.  Let ${\cal Q}=\{Q_n, n \in \N\}$. 
 Let $P=\sum_{n \in \N} \alpha_n  Q_n$. Let $X \in L^1(c_p), X \geq 0$, i.e. from Lemma \ref{lemmaorder}, $X=|X|$. $E_{P}(X)=0$ if and only if $E_{Q_n}(|X|)=0$ for all $n \in \N$.\\
From equation (\ref{eqnum0}), it follows that for $X \geq 0$, $E_{P}(X)=0$ if and only if $c_p(X)=0$ if and only if $X=0$ in $L^1(c_p)$.\\This proves that the canonical $c_p$-class is well defined (as it is not empty) and that $\sum_{n \in \N} \alpha_n  Q_n$ belongs to the canonical $c_p$-class.
\end{proof} 

\begin{lemma}  Let $P$ be a probability measure  belonging to the canonical $c_p$-class. Let $X $ be an element of $L^1(c_p)$. Then $X \geq 0$ (for the order in $L^1(c_p)$) if and only  $X \geq 0 \; P \;a.s.$ 
\label{lemmaorder2}
\end{lemma}
\begin{proof}
For every $X \in L^1(c_p)$, $|X|-X \geq 0$. From Lemma \ref{lemmaorder} $X \geq 0$  if and only if $|X|-X=0$ in $L^1(c_p)$. By definition of the canonical $c_p$-class this is equivalent to  $|X|-X=0\; \; P \;$ a.s., i.e. $X \geq 0\;\; P$ a.s.
\end{proof}
\begin{remark}
When ${\cal P}=\{P\}$ the canonical $c_p$-class  is the restriction to ${\cal M}^+(c_p)$ of the usual equivalence class of the probability measure $P$.\\
When ${\cal P}$ is a finite set, ${\cal P}=\{P_1,...P_n\}$ the canonical $c_p$-class is the restriction to ${\cal M}^+(c_p)$ of the equivalence class (in the usual sense) of the  probability measure $P=\frac{\sum_{1 \leq i \leq n}P_i}{n}$.
\end{remark}  

Our next goal is to give a description of  $L^1(c_p)^*$.

\begin{theorem}
 There is a regular probability measure $P$ belonging to  the canonical $c_p$-class, and a countable  subset ${\cal D}=\{L_n, \; n \in \N \}$ of  the set $L^1(c_p)^*_+$ of non negative continuous linear forms on $L^1(c_p)$ such that 
\begin{itemize}
\item
$\{L_n, \; n \in \N \}$ is dense in   $L^1(c_p)^*_+={\cal M}^+(c_p)$ for the weak* topology.
\item
 Every $L_n$ is represented by a   non negative measure on $(\Omega,{\cal B}(\Omega))$   absolutely continuous with respect to $P$.
\end{itemize}
Every continuous linear form $\Phi$ on $L^1(c_p)$ is the weak* limit of a sequence $\Phi_n$ where every $\Phi_n$ is the difference of two elements of ${\cal D}$.\\
Furthermore for every $X \geq 0$ in $L^1(c_p)$, $X=0$ iff $P(X)=0$, iff $L_n(X)=0$ for all $n \in \N$.
 \label{thmreplin}
\end{theorem} 
\begin{proof}
Denote $nK_+ =\{L \in L^1(c_p)^*, L \geq 0 \; \text{and} \; ||L|| \leq n\}$. From Corollary \ref{corollaryPolish}, every $nK_+$ is metrizable compact for the weak* topology. There is then in $nK_+$ a dense countable set ${\cal D}_n$. Thus ${\cal D}=\cup _{n \in \N} {\cal D}_n$ is countable and dense in $L^1(c_p)^*_+$ for the weak* topology. Enumerate the elements of ${\cal D}$, ${\cal D}= \{L_n, \; n \in \N\}$. From Proposition \ref{thmreg},  every $L_n$ is represented by a non negative finite measure $\mu_n$ on $(\Omega,{\cal B}(\Omega))$. Let $\alpha_n>0$ such that $\sum_{n \in \N}\alpha_n ||L_n|| <\infty$. Then  $\tilde L= \sum_{n \in \N}\alpha_n  L_n \in L^1(c_p)^*_+$. From Proposition \ref{thmreg}, $\tilde L$ is represented by a non negative finite measure $\mu$. Denote $P$ the probability measure $P=\frac {\mu}{\mu(\Omega)}$. $P$ is a probability measure on $(\Omega,{\cal B}(\Omega))$, $P \in {\cal M}^+(c_p)$.  Furthermore  every $\mu_n$ is absolutely continuous with respect to $P$, and $P$ is regular from  Theorem 1.1 of \cite{Bil}. \\
We  prove now that $P$ belongs to the  canonical $c_p$-class. Every $L_n$ belongs to $L^1(c_p)^*$. Thus for every $X$ in $L^1(c_p)$ such that $X=0$ in
$L^1(c_p)$, $L_n(X)=0$ and thus $\tilde L(X)=0$. It follows that $P(X)=0$. Conversely let $X \geq 0$ in $L^1(c_p)$ such that $P(X)=0$. It follows that $\tilde L(X)=0$. Every $L_n$ belongs to $ L^1(c_p)^*_+$, and $X \geq 0$, thus $L_n(X) \geq 0$ for all $n$. From the equality  $\tilde L(X)=0$, it follows that  $L_n(X)=0 \;\forall n \in \N$.   $\{L_n,\; N \in\N\}$ is dense in  $L^1(c_p)^*_+$ for the weak* topology, therefore $L(X)=0$  for all $L \in L^1(c_p)^*_+$. From the representation result of continuous linear forms on $L^1(c_p)$ (Proposition \ref{thmreg}) and the Jordan decomposition of bounded signed measures on $(\Omega,{\cal B}(\Omega))$, it follows that every $\Phi \in L^1(c_p)^*$ is represented by a bounded measure $\mu=\mu^+-\mu^-$. There is a Borelian set $A$ such that $\int f d \mu^+=\int f1_A d \mu$ for every $f \in {\cal C}_b(\Omega)$. $|\mu|=\mu^++\mu^-$ is defined on $(\Omega,{\cal B}(\Omega))$ and is thus regular from Theorem 1.1 of \cite{Bil}. 
\begin{equation}
\forall \epsilon >0, \; \exists  V open ,\; A \subset V \; such \; that\; |\mu|(1_V-1_A) \leq \frac{\epsilon}{2}
\label{eqno1}
\end{equation}
$1_V$  is lower semi-continuous so it is the increasing limit of a sequence of continuous functions $h_n$. From the monotone convergence theorem, and equation (\ref{eqno1}), it follows that 

\begin{equation}
\forall \epsilon >0, \; \exists h \in {\cal C}_b(\Omega),\; 0 \leq h \leq 1_V, such\; that \; \int |1_A-h|d|\mu| < \epsilon 
\label{eqno2}
\end{equation}
Thus 
\begin{equation}
 |\int f1_A d\mu- \int fh d\mu| < ||f||_{\infty} \epsilon 
\label{eqno3}
\end{equation}
By definition of $\mu$, 
\begin{equation}
\forall f \in {\cal C}_b(\Omega),  |\int fh d\mu| < ||\Phi|| c_p(fh) \leq ||\Phi|| c_p(f)  
\label{eqno4}
\end{equation}
From (\ref{eqno3}) and (\ref{eqno4}), we get $|\int f d\mu^+|=|\int f1_A d\mu| \leq ||\Phi|| c_p(f)$.  
 It follows that $\mu^+$ defines an element of $L^1(c_p)^*_+$. It is the same for $\mu^-$. Thus  for every $\Phi \in L^1(c_p)^*$, $\Phi(X)=0$. From Hahn Banach Theorem, it follows that $X=0$ in $L^1(c_p)$. This proves that $P$ belongs to the  canonical $c_p$-class.\\
We have  proved that every $\Phi \in L^1(c_p)^*$ can be written $\Phi=\Phi^+-\Phi^-$, $\Phi^+, \; \Phi^- \in L^1(c_p)^*_+$. The result follows then from the density of ${\cal D}$ in $L^1(c_p)^*_+$. 
\end{proof}
The results of the previous section on convex risk measures on $L^1(c)$ can be specified when the capacity is $c_p=c_{p,{\cal P}}$. 
\begin{proposition}  Let $\rho$ be a convex risk measure on $L^1(c_{p})$. There is a probability measure $Q$ in the canonical $c_{p}$-class and a countable set $\{Q_n, n \in \N\}$ of probability measures all absolutely continuous with respect to $Q$ such that 
\begin{equation}
\rho(X)=\sup_{n \in \N}[E_{Q_n}(-X)-\alpha(Q_n)]\;\;\; \forall X \in L^1(c_{p})
\label{eqcpsub}
\end{equation}
\label{corcpsub}
\end{proposition}
\begin{proof}
From Theorem \ref{thmrepnum},  there is a countable set $\{Q_n, n \in \N\}$ of probability measures such that equation (\ref{eqcpsub}) is satisfied. From Theorem \ref{proppolar} there is a probability measure $P$ in the canonical $c_{p}$-class. Let $Q =\frac{P}{2}+\sum_{n \in \N} \frac{Q_n}{2^{n+2}}$. It is easy to verify that $Q$ satisfies the required conditions.  
\end{proof}

\begin{remark}
Even if the  
capacity $c_{p}$ is   defined from a  weakly relatively compact set of probability measures, the set of probability measures $\{Q_n,\;n \in \N\}$ in the above dual representation (\ref{eqcpsub}) of a convex risk measure $\rho$    on $L^1(c_{p})$ is not always  
relatively  compact for the weak* topology. From Proposition \ref{prop6}, 
$\{Q_n,\;n \in \N\}$ is relatively compact iff $\rho$ is majorized by a sublinear risk measure.
\end{remark}

\section{Regular risk measures on ${\cal C}_b(\Omega)$}
\label{secreg}
\subsection{Regularity}
Notice that in a context of uncertainty, which is when no reference probability measure is given, it is natural to consider risk measures defined on the space ${\cal C}_b(\Omega)$ or more generally on a lattice vector subspace of ${\cal C}_b(\Omega)$. As in Section \ref{sec1.1}, ${\cal L}$ denotes a linear vector subspace of ${\cal C}_b(\Omega)$ containing the constants, generating the topology of $\Omega$ and which is a vector lattice.        
\begin{definition}
$\rho:{\cal L} \rightarrow \R$ is a convex risk measure on ${\cal L}$ if it satisfies the axioms of  Definition \ref{defmes}, replacing everywhere $L^1(c)$ by ${\cal L}$.  It is normalized if $\rho(0)=0$.
\begin{itemize}
\item A sublinear risk measure $\rho$ on ${\cal L}$  is regular if for every decreasing sequence $X_n$ of elements of ${\cal L}$ with limit $0$, $\rho(-X_n)$ tends to $0$.
\item A normalized convex risk measure is uniformly regular  if for all $X$ $\sup_{\lambda>0}\frac{\rho(\lambda X)}{\lambda} < \infty$, and for every decreasing sequence $X_n$ of elements of ${\cal L}$ with limit $0$, $\frac{\rho(- \lambda X_n)}{\lambda}$ converges to $0$ uniformly in $\lambda$.
\end{itemize}
 \label{defreg}
\end{definition}
\begin{remark} For sublinear risk measures, the two notions of regularity and uniform regularity are equivalent.
\end{remark}
From now on in this section $\rho$ is a normalized convex risk measure on ${\cal L}$.
\begin{lemma} Assume that $\rho$ is uniformly regular. $\rho_{min}(X) = \sup_{\lambda>0}\frac {\rho(\lambda X)}{\lambda}$ defines a regular sublinear risk measure on ${\cal L}$. It is the minimal sublinear risk measure on ${\cal L}$ majorizing $\rho$. 
\label{min}
\end{lemma}
\begin{proof} The convexity, monotonicity and translation invariance of $\rho_{min}$ follow easily from the same  properties of $\rho$. The homogeneity of $\rho_{min}$ follows from its definition. Thus $\rho_{min}$ is a sublinear risk measure on ${\cal L}$ majorizing $\rho$. The regularity of $\rho_{min}$ follows from the uniform regularity of $\rho$. For every sublinear risk measure  $\rho_1$ majorizing $\rho$, for every $X \in {\cal L}$,  $\rho_{min}(X) \leq \rho_1(X)$. Thus $\rho_{min}$ is minimal. 
\end{proof}

\begin{lemma}  For every $Y$ in ${\cal L}$, for every sequence $\lambda_n$ of real numbers decreasing to $1$, the sequence $\rho(\lambda_n Y)$ converges to  the limit $\rho(Y)$.
\label{lemma3.0}
\end{lemma} 
\begin{proof} 
As $\lambda_n$ is a decreasing sequence  with  limit $1$, one can assume that 
$2> \lambda_n \geq 1$. Write $\lambda_n=1+\epsilon_n$, $0 \leq \epsilon_n<1$. From the  convexity of $\rho$ and $\rho(0)=0$, it follows that  
\begin{equation}
\rho((1+\epsilon_n) Y) \geq (1+\epsilon_n) \rho(Y)
\label{eqreg1}
\end{equation}
 $(1+\epsilon_n)Y=(1-\epsilon_n)Y+\epsilon_n (2Y)$. Using the convexity of $\rho$, it follows that
\begin{equation}
\rho((1+\epsilon_n)Y) \leq (1-\epsilon_n)\rho(Y)+\epsilon_n \rho(2Y)
\label{eqreg2}
\end{equation}
From inequations (\ref{eqreg1}) and (\ref{eqreg2}), 
\begin{equation}
(1+\epsilon_n)\rho(Y) \leq \rho((1+\epsilon_n)Y) \leq (1-\epsilon_n)\rho(Y)+ \epsilon_n \rho(2Y)
\label{eqreg3}
\end{equation}
Passing now to the limit in inequality (\ref{eqreg3}), it follows that the sequence $\rho((1+\epsilon_n)Y)$ has a limit equal to $\rho(Y)$.\\
\end{proof} 
Using the preceding Lemma, we prove now that every normalized uniformly regular convex risk measure can be extended into a convex risk measure on $L^1(c)$ for some capacity $c$. Therefore we will be able to apply the representation results of Section  \ref{L1}. 
 
\begin{lemma}
 Assume that $\rho$ is uniformly regular. Denote $\rho_1$ a regular sublinear  risk measure on ${\cal L}$ such that $\rho \leq \rho_1$.
\begin{itemize}
\item
$c(X)= \rho_1(-|X|)$ defines a capacity on ${\cal L}$.
\item
$\rho_1$ has a unique continuous extension into a sublinear risk measure 
$\overline {\rho_1}$  on $L^1(c)$.
\item
$\rho$ has a unique continuous extension into a normalized convex risk measure 
$\overline \rho$  on $L^1(c)$ majorized by  $\overline \rho_1$.
\end{itemize}
\label{Lemmaextend}
\end{lemma}
\begin{proof}
\begin{itemize}
\item The sublinearity, monotonicity and regularity of $\rho_1$ imply that $c$ is a capacity on ${\cal L}$. As usual, this leads to the Banach space $L^1(c)$.
\item
As $\rho_1$ is sublinear, for every $X,Y \in {\cal L}$,
$\rho_1(X) \leq \rho_1(Y) + \rho_1(X-Y)$.\\ Exchanging  $X$ and $Y$ and using the monotonicity of $\rho_1$ and the definition of $c$, it follows that $|\rho_1(X)-\rho_1(Y)| \leq  c(X-Y)$.
Thus $\rho_1$ is uniformly continuous on ${\cal L}$ for the $c$ semi-norm. It extends uniquely into a continuous function $\overline {\rho_1}$ on $L^1(c)$. $\overline {\rho_1}$ is a sublinear risk measure.
\item 
let $\epsilon_n>0$ decreasing to $0$.
$$X=\frac {1}{1+\epsilon_n}[(1+\epsilon_n)Y] + \frac {\epsilon_n}{1+\epsilon_n}[ \frac {1+\epsilon_n}{\epsilon_n}(X-Y)]$$
From the convexity of $\rho$, the majoration of $\rho$ by $\rho_1$ and the homogeneity of $\rho_1$ (cf $\rho_1$ is sublinear), it follows that 
\begin{equation}
\rho(X) \leq \frac {1}{1+\epsilon_n}\rho((1+\epsilon_n)Y) + \rho_1(X-Y)
\label{contreg}
\end{equation}
From inequation (\ref{contreg}) and Lemma \ref{lemma3.0} applied with  $(1+\epsilon_n)Y$, passing to the limit,  it follows then that $\rho(X)-\rho(Y) \leq \rho_1(X-Y) \leq c(X-Y)$. Exchanging $X$ and $Y$,  this proves the uniform continuity of $\rho$ for the $c$ semi-norm. $\rho$ extends then uniquely into a continuous function $\overline {\rho}$ on $L^1(c)$. $\overline {\rho}$ is a convex risk measure on $L^1(c)$ majorized by $\overline {\rho_1}$.
\end{itemize}
\end{proof}

\begin{definition}
  Let $\rho$ be a normalized uniformly regular convex risk measure on ${\cal L}$. The capacity $c_{\rho}$ defined as $c_{\rho}(X)=\rho_{min}(-|X|)$ is called the capacity canonically associated with $\rho$.
\label{defro}
\end{definition}
\subsection{Representation of uniformly regular convex risk measures}
In this section, we assume that $\Omega$ is a  Polish space.  
Taking into account the liquidity risk in a financial market, we introduce the following definition for a  riskless asset, which means that all investment in this asset is risk-free.
\begin{definition} A non positive element $X$ of ${\cal C}_b(\Omega)$ is riskless if for all $\lambda>0$, $\rho(\lambda X)=0$ (or equivalently for all  $\lambda>0$, $\rho(\lambda X) \leq 0$).
\label{defrisk}
\end{definition}
\begin{theorem}
Let $\rho$ be a normalized uniformly  regular convex risk measure on ${\cal L}$.\\ 
Then $\rho$ extends uniquely to ${\cal C}_b(\Omega)$ and  admits the following representation 
\begin{equation}
\forall X \in{\cal C}_b(\Omega)\;\;\rho(X)= \sup_{n \in \N} (E_{Q_n}(-X)-\alpha(Q_n))
\label{eqatt1}
\end{equation} 
for a certain  weakly relatively compact set  $\{Q_n, \; n \in \N\}$ of probability measures.
Furthermore for $\alpha_n>0$ such that $\sum_{n \in \N} \alpha_n=1$  the probability measure $P=\sum_{n \in \N} \alpha_n Q_n$ characterizes the riskless non negative elements of ${\cal C}_b(\Omega)$, that is $X \leq 0$ is riskless iff $X=0$ P a.s.\\
 For every $X \in {\cal C}_b(\Omega)$, there is a probability measure $Q_X$ in the weak closure of $\{Q_n, \; n \in \N\}$,  such that 
\begin{equation}
\rho(X)=  E_{Q_X}(-X)-\alpha(Q_X)
\label{eqatt2}
\end{equation}

\label{thmerepcont}
\end{theorem}
\begin{proof} Let $c_{\rho}(X)=\rho_{min}(-|X|)$ be the capacity canonically associated with $\rho$ (definition \ref{defro}). As $\Omega$ is a Polish space, every capacity is a Prokhorov capacity.
Denote  $\rho$  (resp $\rho_{min}$) the extensions of $\rho$  (resp $\rho_{min}$) to $L^1(c_{\rho})$ given by  Lemma \ref{Lemmaextend}. \\
As  $\rho$ is majorized by  $\rho_{min}$, the representation result with a countable weakly relatively compact set ${\cal Q}=\{Q_n\}$  follows from Proposition  \ref{proprepnum}. We can of course restrict to $Q_n$ such that $\alpha(Q_n) <\infty$. Then $c_{\rho}(X)=\sup_{n \in \N} E_{Q_n}(|X|)$ i.e. $c_{\rho}=c_{1,{\cal Q}}$. From Theorem \ref{proppolar}  the probability measure $P=\sum_{n \in \N} \alpha_n Q_n$ belongs to the canonical  $c_{\rho}$-class. Let $X \leq 0$ in ${\cal C}_b(\Omega)$, $X$ is riskless iff $\rho(\lambda X)=0 \;\forall \lambda >0$, iff $c_{\rho}(-X)=0$, iff $X=0$ $P$ a.s. The existence of $Q_X$ follows from Theorem \ref{thmmax}. 
\end{proof}

\section{Examples}\label{Exemples} 
\subsection{ G-expectations}
In all this section, $\Omega={\cal C}_0 ([0,\infty[,\R^d)$,    the set of continuous functions $f$ defined on $[0,\infty[$ with values in $\R^d$ such that $f(0)=0$. ${\cal C}_0([0,\infty[,\R^d)$ endowed with the topology of uniform convergence on compact spaces  is a Polish space. \\
Peng introduced the notion of sublinear expectation and of  G-expectations \cite{P1} \cite{P2} defined on a vector lattice ${\cal H}$ of real functions containing $1$ and included in ${\cal C}_b(\Omega)$. 
For the definition of a sublinear expectation $\E$ on ${\cal H}$ we refer to \cite{DHP} section 3. 
G-expectations are defined from solutions of P.D.E. in  \cite{P1} and \cite{P2}. A G-expectation is  up to a minus sign a sublinear risk measure.\\
It is proved in  \cite{DHP} and \cite{HP} that every $G$-expectation $\E$ has a  representation  with respect to a weakly relatively compact set of probability measures ${\cal P}$: 
$\E(f)=\sup_{P \in {\cal P}}E_P(f)\;\;$ for all $f$ in ${\cal H}$.
$\E$ extends naturally to ${\cal C}_b(\Omega)$: 
\begin{equation}
\E(f)=\sup_{P \in {\cal P}}E_P(f)\;\; \forall f \in{\cal C}_b(\Omega)
\label{E} 
\end{equation}
As ${\cal P}$ is weakly relatively compact, $\rho(f)=\E(-f)$ is a sublinear regular risk measure on ${\cal C}_b(\Omega)$. Denote $c_{\E}=c_{\rho}$ the corresponding capacity $c_{\E}(X)=\E(|X|)\;\forall X \in {\cal C}_b(\Omega)$.\\
Notice that alternatively, regularity could be proved directly for G-expectations. Theorem \ref{thmerepcont} would thus give the representation result (equation \ref{E}).
 \begin{proposition}
  There is a countable weakly relatively compact set  $\{Q_n, n \in \N\}$ of probability measures, $Q_n  \in {\cal P}$  such that 
\begin{equation}
\;\forall X \in {\cal C}_b(\Omega)\;\;\E(X)=\sup_{ n \in \N} E_{Q_n}(X)
\label{eqexp}
\end{equation}
Let $P=\sum_{n \in \N^*} \frac{Q_n}{2^{n+1}}$. For all $f \geq 0$ in ${\cal C}_b(\Omega)$, $\E(f)=0$ iff $f=0$ $P$ a.s.\\
  For every $X \in {\cal C}_b(\Omega)$,  there is a probability measure $Q_X$ in the weak closure of $\{Q_n, n \in \N^*\}$, such that  
$\E(X)=E_{Q_X}(X)$. 
\label{prop9}
\end{proposition}
\begin{proof}
 The result follows  from  Theorem \ref{thmerepcont}.
\end{proof}
\subsection{Risk measure in context of uncertain volatility} 
\label{secDM}
We consider a framework introduced in \cite{DM}.  
Let $\Omega={\cal C}_0([0,T],\R^{d})$ the space of continuous functions on $[0,T]$ null in zero. For every $t \leq T$, let $\Omega_t= {\cal C}_0([0,t],\R^{d})$. $\Omega_t$ is identified with the subset of $\Omega$ of elements which are constant on $[t,T]$. Let ${\cal B}_t$ be the $\sigma$-algebra on $\Omega$ generated by the open sets of $\Omega_t$. Denote $B_t$ the coordinate process. A probability measure $Q$ on $(\Omega, {\cal B}(\Omega))$ is called an orthogonal martingale measure if the coordinate process $(B_t)$ is a martingale with respect to ${\cal B}_t$ under $Q$ and if the martingales
$(( B_{i})_t)_{1\leq i\leq d}$ are orthogonal in
the sense that for all $i\neq j$, $ <B_{i},B_{j}>^Q_t=0\ \ Q\ a.s. $. $<B_{i}, B_j>^Q$ denotes the quadratic covariational process
corresponding to  $B^{i}$ and $B^j$, under $Q$ and $<B>^Q$ the quadratic variation of $B$ under $Q$. Fix for all $i\in\left\{ 1,\dots,d\right\}$ two finite deterministic H\"{o}lder-continuous measures $\underline{\mu}_{i}$ and $\mu_{i}$ on $[0,T]$  and consider the set ${\cal P}$ of orthogonal martingale measures such that $$\forall i\in\{1,\dots,d\},\ \ \   d\underline{\mu}_{i,t} \leq d<B_{i}>_{t}^{Q}
\leq d\mu_{i,t}.$$
M. Kervarec has proved in \cite{K}, Lemma 1.3  that the set ${\cal P}$ is weakly relatively compact. Thus $c_1(f)= \sup_{Q \in {\cal P}} E_Q(|f|)$  defines a capacity on ${\cal C}_b(\Omega)$ (see Appendix, Section \ref{secap}).
  As in Section \ref{sec2}, $L^1(c_1)$ denotes the corresponding Banach space, containing  ${\cal C}_b(\Omega)$ as a dense subset.  From Theorem \ref{Thm1}, and Theorem \ref{proppolar}, there is  a countable set  $(P_n)_{n \in \N}$, $P_n \in {\cal P}$  such that
$\forall X \in L^1(c_1)$, $c_1(X)=\sup_{n \in \N} E_{P_n}(|X|)$ and such that $P= \sum_{n \in \N} \frac{P_n}{2^n}$  belongs to the canonical $c_1$-class.

\begin{lemma}
For every probability measure $R$ defining an element of $L^1(c_1)^*$, 
$$\forall i\in\{1,\dots,d\},\ \ \   d\underline{\mu}_{i,t} \leq d<B_{i}>_{t}^{R}
\leq d\mu_{i,t}.$$
\label{lemmaUN}
\end{lemma}
Notice that a probability measure $R$ in $L^1(c)^*$ does not necessarily belongs to ${\cal P}$ and therefore the result is not trivial.
\begin{proof}
From \cite{DM}, $(B_i)_s^2 \in L^1(c_1)$ for every $s$, thus $\int_0^t (B_i)_s d(B_i)_s$ can be defined as an element of $L^1(c_1)$. We thus define the quadratic variation of $B$ in $L^1(c_1)$ by 
\begin{equation}
<B_i>^{c_1}_t=(B_i)_t^2-2 \int_0^t (B_i)_s d(B_i)_s
\label{qv}
\end{equation}
This equation is satisfied in $L^1(c_1)$ thus it is satisfied $R \; a.s.$ for every probability measure $R$ defining an element of $L^1(c_1)^*$. 
 Let $s \leq t$. Let $A =\{\omega\;|\; <B_i>^{c_1}_t-<B_i>^{c_1}_s > 
{\mu}_{i} [s,t]\} \cup \{\omega\;|\; <B_i>^{c_1}_t-<B_i>^{c_1}_s <  \underline{\mu}_{i} [s,t]\}$.  By hypothesis $P_n(A)=0$. Thus $P(A)=0$. 
The inequality 
\begin{equation}
{\mu}_{i} [s,t] \geq <B_i>^{c_1}_t-<B_i>^{c_1}_s \geq \underline{\mu}_{i} [s,t]
\label{eqinqv}
\end{equation}
is thus satisfied $ P $ a.s.
From Lemma \ref{lemmaorder2}, inequality (\ref{eqinqv}) is then satisfied in $L^1(c_1)$ and then also  $R \;a.s.$ for every probability measure defining an element of  $L^1(c_1)^*$.

\end{proof}
\begin{proposition}
The set ${\cal P}$ is convex metrizable compact for the weak* topology $\sigma(L^1(c_1)^*, L^1(c_1))$ and also for the weak topology.
\label{propcompun}
\end{proposition}
\begin{proof}
The convexity of ${\cal P}$ is obvious.
Denote as in Section \ref{sec1},  $K_+$ the non negative part of the unit ball of $L^1(c)^*$. From the  definition of $c_1$ it follows that  ${\cal P} \subset K_+$. Thus the weak*closure $\overline {\cal P}$ of ${\cal P}$ is a subset of $K_+$. From Lemma \ref{lemmaUN} it follows that every element $Q \in \overline {\cal P}$ satisfies $$\forall i\in\{1,\dots,d\},\ \ \   d\underline{\mu}_{i,t} \leq d<B_{i}>_{t}^{Q}
\leq d\mu_{i,t}$$ From Corollary \ref{corollaryPolish}, $K_+$ is metrizable compact for the weak* topology thus for every $Q \in \overline {\cal P}$, there is a sequence $Q_n, Q_n \in{\cal P}$ converging to $Q$ for the weak* topology.\\ From \cite{DM}, $|(B_i)_t|^k \in L^1(c_1)$ for $k=1$ or $2$, so  
$(E_{Q_n}-E_{Q})(|(B_i)_t |^k) \rightarrow 0$.
 Passing to the limit, $E_{Q}|(|(B_i)_t|) \leq c_1(|(B_i)_t|)$ and 
\begin{equation}
E_{Q}|(|(B_i)_t|^2) \leq c_1(|(B_i)_t^2|)
\label{eqBL2}
\end{equation}  
Let $ g$  in ${\cal C}_b(\Omega_s)$. $g$ can be identified with the element $\tilde g$ of ${\cal C}_b(\Omega)$ defined by $\tilde g (x)=g(x_{|[0,s]})$.    It follows from the inequality $c_1(Xg) \leq ||g||_{\infty} c_1(|X|)$ that $\forall u \geq s$, $(B_i)_u g \in L^1(c_1)$, so $\forall g \in {\cal C}_b(\Omega_s)$ $\forall \lambda \in \R$,
\begin{equation}
 (E_{Q_n}-E_{Q})((B_i)_u (g+\lambda)) \rightarrow 0
\label{eqconv}
\end{equation}$(B_i)_t$ is a martingale for $Q_n$, thus passing to the limit in (\ref{eqconv}), with $u=t$ and $u=s$, we obtain $\forall g \in {\cal C}_b(\Omega_s)$ $\forall \lambda \in \R$,
\begin{equation} 
E_{Q}((B_i)_t (g+ \lambda))=E_{Q}((B_i)_s (g+ \lambda)) 
\label{eqconv1}
\end{equation}
 From (\ref{eqBL2}), $(B_i)_u \in L^2(\Omega,{\cal B}_u, Q)$ for $u=t,s$, and $\{g+ \lambda, \;g \in{\cal C}_b(\Omega_s), \; \lambda \in \R\}$ is dense in $L^2(\Omega,{\cal B}_s, Q)$, thus the equality (\ref{eqconv1}) is satisfied for every $g \in L^2(\Omega,{\cal B}_s, Q)$. This proves that $(B_i)_t$ is a martingale for $Q$. A very similar proof leads to the fact that the martingales $(B_i)_t$ and $(B_j)_t$ are mutually orthogonal for $i \neq j$.
Thus ${\cal P}$  is closed for the weak* topology. As ${\cal P} \subset K_+$,  ${\cal P}$  is metrizable compact  for the weak* topology. The result follows  from Proposition \ref{lemmatop}  for the weak topology.
\end{proof}

For every $P \in {\cal P}$ let $\beta(P) \geq 0$. Let $\rho$ be defined by
\begin{equation}
\forall X \in {\cal C}_b(\Omega)\;\;\rho(X)=\sup_{P \in {\cal P}}(E_P(-X)-\beta(P))\;\;
\label{eq_30}
\end{equation}
As ${\cal P}$ is metrizable compact for the weak topology, $\rho -\rho(0)$ is a uniformly regular convex risk measure. Thus Theorem \ref{thmerepcont} applies.\\
The link between the two previous examples is studied in \cite{DHP}. The convex weakly compact set characterizing the G-expectation $\E$ is in fact contained in the set  ${\cal P}$ of orthogonal martingale measures introduced in \cite{DM} and considered in Section \ref{secDM}.
\section{Appendix}
\label{secap}
Let $\Omega$ be a metrizable separable space and $\mathcal{L}$ as in Section 2 a lattice of continuous bounded functions, containing constants and generating the topology of $\Omega$. We now recall some definitions and propositions proved in Section 2 of \cite{FP}.
A capacity is defined as in definition \ref{defcap11}, Section \ref{sec1}.
\begin{definition}
A capacity $c$ defined on ${\cal L}$ is  regular if it satisfies:\\
 For all decreasing net $f_{\alpha}\in\mathcal{L}$ converging to $0$, $\inf c\left( f_{\alpha}\right)=0$.
\end{definition}

\begin{definition}
A capacity $c$ defined on ${\cal L}$ is  a Prokhorov capacity if:\\
 For all $\epsilon >0$, there exists a compact set $K$ such that $c\left( f\right)\leqslant \epsilon$ for all $f\in\mathcal{L}$ such that $\vert f \vert\leq \mathbf{1}_{\Omega \backslash K}$.
\end{definition}

\begin{proposition}
If $\Omega$ is a Lindel\"of space
then every capacity is a regular capacity.
\end{proposition}

\begin{proposition}
If $\Omega$ is locally compact or a Polish space then every regular capacity is a Prokhorov capacity.
\end{proposition}
\begin{remark}
If $\Omega$ is a Polish  space, then it is a Lindel\"of space and thus every capacity is  a Prokhorov capacity. 
\end{remark}
\begin{proposition}
If $\mathcal{P}$ is weakly relatively compact  $c$ defined on ${\cal C}_b(\Omega)$ by $ c(f)=\sup_{P\in\mathcal{P}}(E_P\left[ \left| f\right|^p\right])^{\frac{1}{p}}$ is a  capacity.
\end{proposition}
The proof follows from Dini Theorem (see Proposition I.3 in \cite{K} for more details).

\section*{Acknowledgements}
We thank an anonymous referee for very helpful comments.

\end{document}